\documentclass{amsart}
\usepackage{amsthm,amsfonts,amsmath,amscd,amssymb,latexsym,epsfig}
\usepackage[title,titletoc]{appendix}
\newcommand{\gl}{{\mathfrak g \mathfrak l}}

\newcommand{\g}{{\mathfrak g}}         

\newcommand{\cx}{{\mathbb C}}

\newcommand{\ad}{\operatorname{ad}}
\newcommand{\Ad}{\operatorname{Ad}}
\newcommand{\tr}{\operatorname{tr}}

\newcommand{\re}{\operatorname{Re}}
\newcommand{\im}{\operatorname{Im}}

\newcommand{\Ker}{\operatorname{Ker}}

\newcommand{\Jac}{\operatorname{Jac}}

\newcommand{\grad}{\operatorname{grad}}

\newcommand{\Mat}{\operatorname{Mat}}

\newcommand{\str}{\operatorname{str}}

\numberwithin{equation}{section}

\newtheorem{theorem}{Theorem}[section]

\newtheorem{corollary}[theorem]{Corollary}

\newtheorem{proposition}[theorem]{Proposition}

\theoremstyle{remark}

\newtheorem{remark}[theorem]{Remark}

\newtheorem{definition}[theorem]{Definition}

\newtheorem{example}[theorem]{Example}

\newcommand{\A}{{\mathbb{A}}}
\newcommand{\B}{{\mathbb{B}}}

\newcommand{\oP}{{\mathbb{P}}}

\newcommand{\sF}{{\mathcal{F}}}
\newcommand{\sG}{{\mathcal{G}}}   

\newcommand{\sL}{{\mathcal{L}}}   
\newcommand{\sM}{{\mathcal{M}}}   

\newcommand{\sO}{{\mathcal{O}}}

\newcommand{\fG}{{\mathfrak{g}}}

\newcommand{\fK}{{\mathfrak{k}}}
\newcommand{\fL}{{\mathfrak{l}}}
\newcommand{\fM}{{\mathfrak{m}}}

\newcommand{\fU}{{\mathfrak{u}}}

\begin{document}

\title[Nahm's equations, BHT-equations and Lie superalgebras]{Nahm's, Basu-Harvey-Terashima's equations and Lie superalgebras}
\author{Roger Bielawski}
\address{}


\begin{abstract} We discuss the correspondence between Nahm's equations, the Basu-Harvey-Terashima equations, and Lie superalgebras.
\end{abstract}

\maketitle

\thispagestyle{empty}

\section{Introduction}

This paper arose from the following observation. Consider $\Mat_{n,m}(\cx)\oplus \Mat_{m,n}(\cx)$ with its flat $U(n)\times U(m)$-invariant hyperk\"ahler structure. Let $\mu_1,\mu_2,\mu_3$ be the hyperk\"ahler moment map for the $U(n)$-action and $\nu_1,\nu_2,\nu_3$  the hyperk\"ahler moment map for the $U(m)$-action. Then, along the gradient flow of
\begin{equation}
 F=|\mu_1|^2-|\nu_1|^2, \label{F1}
\end{equation}
$\mu_1,\mu_2,\mu_3$ and $-\nu_1,-\nu_2,-\nu_3$ satisfy Nahm's equations.
\par
This fact has several explanations and consequences. At the simplest level, it follows from the fact  that
\begin{equation}
 I_1(X_{\mu_1}- X_{\nu_1})= I_2(X_{\mu_2}- X_{\nu_2})= I_3(X_{\mu_3}- X_{\nu_3}),
\end{equation}
where $X_\rho$ is the vector field generated by a $\rho$ in the Lie algebra of the symmetry group.
\par
The function \eqref{F1} is a quartic polynomial on $W=\Mat_{n,m}(\cx)\oplus \Mat_{m,n}(\cx)$ and the gradient flow equations are
\begin{equation}
 \begin{matrix}\dot A=\frac{1}{2}(AB B^\ast-B^\ast B A)\\ \dot B=\frac{1}{2}(A^\ast A B-BAA^\ast).
 \end{matrix}
\label{Ter}
\end{equation}
These equations are known as the ABJM version of the Basu-Harvey equations and are due to Terashima \cite{Ter}, and, consequently, we shall refer to them as the {\em BHT-equations}. We observe that they have a very natural interpretation as {\em double superbracket equations} on the odd part of the Lie superalgebra $\gl_{n|m}(\cx)$: 
\begin{equation}
 \dot{C}=\frac{1}{2}[[J(C),C],C]],\label{triple}
\end{equation}
where $C=\begin{pmatrix}0 & A\\ B & 0\end{pmatrix}$ and $J$ is the quaternionic automorphism $J(A,B)=(-B^\ast,A^\ast)$. This equation makes sense for any complex anti-Lie triple system \cite{FoF} equipped with a quaternionic automorphism, and we observe that any solution of \eqref{triple} in this general setting leads to a solution to Nahm's equations (with values in an appropriate Lie algebra).
\par
We give two more interpretations of equations \eqref{Ter}. Firstly, there is a geometric interpretation as a gradient flow on a $GL_n(\cx)\times GL_m(\cx)$-orbit in $\Mat_{n,m}(\cx)\oplus \Mat_{m,n}(\cx)$ for a {\em quadratic} function with respect to certain indefinite metric (\S\ref{Laxs}).
\par
Secondly, similarly to Nahm's equations \cite{Hit1,Hit2}, there is an interpretation as a linear flow on the Jacobian of an algebraic curve. This time, however, the spectral curve is a subscheme of $\oP^2$ and the flow is restricted to line bundles equivariant with respect to certain involution $\tau$ of the spectral curve. 

\section{Moment maps, Nahm's and  the Basu-Harvey-Terashima equations\label{1}} 

We consider the vector space $W_{n,m}=\Mat_{n,m}(\cx)\oplus \Mat_{m,n}(\cx)$ with its natural flat hyperk\"ahler structure: the quaternionic structure $J$ is given by $J(A,B)=(-B^\ast, A^\ast)$ and the metric is 
\begin{equation} \frac{1}{2} \re\tr \left(dA\otimes dA^\ast+dB \otimes dB^\ast\right).\label{metric}\end{equation}
This hyperk\"ahler structure is invariant under the natural $U(n)\times U(m)$-action, given by
\begin{equation}
 (g,h).(A,B)=(gAh^{-1},hBg^{-1}).\label{action}
\end{equation}
The hyperk\"ahler moment map for the $U(n)$-action is
$$ i\mu_1(A,B)=\frac{1}{2}(AA^\ast-B^\ast B),\quad (\mu_2+i\mu_3)(A,B)=AB,$$
while the moment map for the $U(m)$-action is
$$ i\nu_1(A,B)=-\frac{1}{2}(A^\ast A-BB^\ast),\quad (\nu_2+i\nu_3)(A,B)=-BA.$$
Here we identified Lie algebras with their duals using the $\Ad$-invariant metrics $\|X\|^2=-\tr X^2$. A simple calculation shows that for $i=1,2,3$ and 
any $(A,B)\in W_{n,m}$
$$ \|\mu_i(A,B)\|^2-\|\nu_i(A,B)\|^2=\frac{1}{2}\tr(A^\ast ABB^\ast-B^\ast BAA^\ast).$$
The fact that $\|\mu_i\|^2-\|\nu_i\|^2$ is independent of $i$ has the following consequence.
\begin{proposition}
 Let $m(t)\in W_{n,m}$ be a  gradient flow curve of the function $F=\frac{1}{2}\|\mu_1(A,B)\|^2-\frac{1}{2}\|\nu_1(A,B)\|^2$. Then the $\fU(n)$-valued functions $T_i(t)=\mu_i(m(t))$ satisfy Nahm's equations
\begin{equation} \dot{T}_1=[T_2,T_3],\enskip \dot{T}_2=[T_3,T_1],\enskip \dot{T}_3=[T_1,T_2].\label{Nahm}
\end{equation}
Similarly, the $\fU(m)$-valued functions $S_i(t)=-\nu_i(m(t))$ satisfy the Nahm equations.\label{NBHT}
\end{proposition}
\begin{proof} The gradient vector field of $F$ is $I_1X_{\mu_1}-I_1X_{\nu_1}$. Since $F$ is also equal to $\frac{1}{2}\|\mu_i(A,B)\|^2-\frac{1}{2}\|\nu_i(A,B)\|^2$ for $i=2,3$, we obtain
 $$  I_1X_{\mu_1}- I_1X_{\nu_1}= I_2X_{\mu_2}- I_2X_{\nu_2}= I_3X_{\mu_3}- I_3X_{\nu_3}.$$
We compute, using the fact that the moment map for a group action is invariant with respect to any commuting Lie group action,
$$ \dot{T}_1=d\mu_1(\nabla F)=d\mu_1(I_2X_{\mu_2}- I_2X_{\nu_2})=d\mu_1(I_2X_{\mu_2})=d\mu_3(X_{\mu_2})=[\mu_2,\mu_3]=[T_2,T_3],$$
and similarly for $T_2,T_3$. The argument for the $S_i$ is completely analogous.
\end{proof}
The gradient flow equations for the function $ F=\tr \left(A^\ast ABB^\ast-B^\ast BAA^\ast\right)$ are
\begin{equation}
 \begin{matrix}\dot A=\frac{1}{2}(AB B^\ast-B^\ast B A)\\ \dot B=\frac{1}{2}(A^\ast A B-BAA^\ast).
 \end{matrix}
\label{3rd}
\end{equation}
One can also check  directly that, for a solution $A,B$ of these equations, the functions  $T_1=\frac{1}{2i}(AA^\ast-B^\ast B)$, $T_2+iT_3=AB$ satisfy Nahm's equations, and similarly the functions 
$S_1=\frac{1}{2i}(A^\ast A-B B^\ast)$, $S_2+iS_3=BA$.
\par
For the reason mentioned in the introduction, we shall refer to equations \eqref{3rd} as the {\em Basu-Harvey-Terashima (BHT) equations}.
\par
\begin{remark}
Similarly to Nahm's equations, there exists a gauge-dependent version of the BHT-equations. Introduce two more matrix valued functions $u(t)\in \fU(n)$ and $v(t)\in \fU(m)$ and consider the following equations:
\begin{equation}
 \begin{matrix}\dot A +uA-Av=\frac{1}{2}(AB B^\ast-B^\ast B A)\\ \dot B+vB-Bu=\frac{1}{2}(A^\ast A B-BAA^\ast).
 \end{matrix}
\label{gauge}
\end{equation}
These equations are invariant under the following $U(n)\times U(m)$-valued gauge group action:
$$ A\mapsto gAh^{-1},\enskip B\mapsto hBg^{-1},\enskip u\mapsto gug^{-1}-\dot{g}g^{-1},\enskip v\mapsto hvh^{-1}-\dot{h}h^{-1}.
$$
\end{remark}


\section{Lax pair interpretation\label{Laxs}}

\begin{proposition}
  Let $I$ be an interval and $n\geq m$ be two positive integers. Let  $X:I\to \Mat_{n,m}$, $Y:I\to  \Mat_{m,n}$ be of class $C^k$, $k\geq 1$, and of rank $m$ for all $t\in I$. Suppose that $Z=XY$ satisfies the Lax equation $\dot Z=[M,Z]$ for some $M:I\to \Mat_{n,n}$ of class $C^{k-1}$. Then there exists a unique  $N:I\to \Mat_{m,m}$ of class $C^{k-1}$, such that the following equations are satisfied:
\begin{equation}
\begin{matrix}\dot X=MX+XN\\ \dot Y=-YM-NY.\end{matrix}\label{first}
\end{equation}
Conversely, if \eqref{first} are satisfied, then $\dot Z=[M,Z]$, and, moreover, $W=YX$ satisfies the Lax equation $\dot W=[W,N]$.
\end{proposition}
\begin{proof} A direct computation shows that if \eqref{first} are satisfied, then both $Z$ and $W$ satisfy the relevant Lax equations. Conversely, suppose that $\dot Z=[Z,M]$, i.e. $\dot X Y+ X \dot Y-MXY+XYM=0$. We rewrite this as
$$ (\dot X-MX)Y+X(\dot Y+YM)=0.$$
Let $U=\dot X-MX$ and $V=-\dot Y-YM$. Then $UY=XV$. Since $X$ and $Y$ have rank $m$, there are unique $N_1,N_2$ such that $U=XN_1,V=N_2Y$. It follows that $X(N_1-N_2)Y=0$, and using the maximality of the rank of $X,Y$, we conclude that $N_1=N_2$. The equations \eqref{first} are satisfied with $N=N_1=N_2$. The differentiability class of $N$ follows from its uniqueness.
\end{proof}

Equations \eqref{first} can also be written in the Lax form. Set
\begin{equation*}C=\begin{pmatrix} 0 & X\\ Y & 0\end{pmatrix},\quad C_+= \begin{pmatrix} M & 0\\ 0 & -N\end{pmatrix}.\label{CC}\end{equation*}
Then \eqref{first} is equivalent to
\begin{equation} \dot{C}=[C_+,C].\label{dotC}\end{equation}

We now introduce a spectral parameter, and consider $X(\zeta)=A_0+A_1\zeta$, $Y(\zeta)=B_0+B_1\zeta$. We suppose that $Z(\zeta)=X(\zeta)Y(\zeta)$ satisfies the Lax equation
$$ \dot Z=[Z_\#,Z],$$
where $Z_\#=\frac{1}{2}(A_0B_1+A_1B_0)+A_1B_1\zeta$. According to the previous proposition, we can find an $N(\zeta)$, so that \eqref{first} holds. We have then $\dot W=[-N,W]$, where $W(\zeta)=Y(\zeta)X(\zeta)$. We try $N(\zeta)$ of the form $N=-W_\#$, i.e. $-N=\frac{1}{2}(B_0A_1+B_1A_0)+B_1A_1\zeta$.
Substituting into the equations \eqref{first} we obtain:
\begin{equation}
 \begin{matrix}\dot A_0=\frac{1}{2}(A_1B_0A_0-A_0B_0A_1)\\ \dot A_1=\frac{1}{2}(A_1B_1A_0-A_0B_1A_1)\\ \dot B_0=\frac{1}{2}(B_1A_0B_0-B_0A_0B_1)\\
  \dot B_1=\frac{1}{2}(B_1A_1B_0-B_0A_1B_1).
 \end{matrix}
\label{2nd}
\end{equation}
Thus, these equations are equivalent to
\begin{equation}
  \dot Z=[Z_\#,Z],\quad  \dot W=[W_\#,W], \label{Lax}
\end{equation}
where $Z,Z_\#,W,W_\#$ are defined above. Now suppose that the $A_i$ and $B_i$ satisfy the reality condition: $A_1=-B_0^\ast$, $B_1=A_0^\ast$. We write simply $A,B$ for $A_0,B_0$. It follows that the equations \eqref{Lax} are equivalent to Nahm's equations for $(T_1,T_2,T_3)$ and $(S_1,S_2,S_3)$, where
\begin{equation}T_2+iT_3=AB,\quad iT_1=\frac{1}{2}(AA^\ast-B^\ast B),\label{N1}\end{equation}
\begin{equation} S_2+iS_3=BA,\quad iS_1=\frac{1}{2}(A^\ast A-B B^\ast).\label{N2}\end{equation}
 On the other hand, \eqref{2nd} becomes the BHT-equations \eqref{3rd}.
Therefore we have a different proof of Proposition \ref{NBHT}, the statement of which can be strengthened as follows:
\begin{corollary}
 Let $n\geq m$. The equations \eqref{3rd} are equivalent to Nahm's equations for $(T_1,T_2,T_3)$ defined by \eqref{N1}. In addition, they imply Nahm's equations for
$(S_1,S_2,S_3)$ defined by \eqref{N2}.\hfill $\Box$\label{AB}
\end{corollary}

We now rewrite \eqref{3rd} in the form \eqref{dotC}. Thus
$$ C=\begin{pmatrix} 0 & A-B^\ast \zeta\\ B+A^\ast\zeta & 0\end{pmatrix}, $$
$$ C_+= \begin{pmatrix} \frac{1}{2}(AA^\ast-B^\ast B)-B^\ast A^\ast\zeta & 0 \\ 0& \frac{1}{2}(A^\ast A-BB^\ast)- A^\ast B^\ast\zeta \end{pmatrix}.$$
Let us write $\sM_0$ for the block-diagonal part of $\Mat_{m+n,m+n}$ and $\sM_1$ for the off-diagonal part. Thus $C\in \sM_1$ and $C_+\in \sM_0$. Let $J:\sM_1\to \sM_1$ be the canonical quaternionic structure on $\sM_1=\Mat_{n,m}\oplus \Mat_{m,n}$, i.e. 
$$ J\begin{pmatrix} 0 & A\\B & 0\end{pmatrix}=\begin{pmatrix} 0 & -B^\ast \\A^\ast & 0\end{pmatrix}.$$
We have $J^2=-1$ and equations \eqref{3rd} can be written as
\begin{equation} \dot{C}=\frac{1}{2}[CJ(C)+J(C)C,C]=\frac{1}{2}[J(C),C^2].\label{JC}\end{equation}

This equation implies that the flow of $(A,B)$ remains in an orbit of $GL_n(\cx)\times GL_m(\cx)$ (in fact $S(GL_n(\cx)\times GL_m(\cx))$ and $SL_n(\cx)\times SL_m(\cx)$ for $n=m$). We shall now interpret the flow as a gradient flow on such an orbit, with respect to certain metric.
First of all, let us view $\Mat_{m+n,m+n}$ as the Lie superalgebra  $\fG\fL_{n|m}(\cx)$ with $\sM_0$ being the even part and $\sM_1$ the odd part. The Lie superbracket is defined as  $[A,B]=AB-(-1)^{|A||B|}BA$. Equation \eqref{JC} can be then written as
\begin{equation} \dot{C}=\frac{1}{2}[[J(C),C],C].\label{sJC}
\end{equation}
\begin{remark}
The map $J$ is the restriction of the following antilinear map on $\fG\fL_{n|m}(\cx)$:
\begin{equation} \begin{pmatrix} U & A\\B & V\end{pmatrix}\mapsto \begin{pmatrix} -U^\ast & -B^\ast \\A^\ast & -V^\ast\end{pmatrix},\label{JJ}\end{equation}
which we also denote by $J$. It is the negative of complex conjugation followed by the supertranspose, and, hence, it commutes with the superbracket. One could therefore consider equation \eqref{sJC} on all of $\fG\fL_{n|m}(\cx)$, rather than just on the odd part.
\end{remark}

\medskip

Recall now the notion of the supertrace:
$$\str \begin{pmatrix} U & A\\B & V\end{pmatrix}=\tr U-\tr V.$$
It has the following $\ad$-invariance property:
\begin{equation}\str [X,Y]Z+(-1)^{|X||Y|}\str Y[X,Z]=0.\label{ad-inv}\end{equation}

We define the following symmetric form on $\fG\fL_{n|m}(\cx)$::
\begin{equation} \langle X, Y\rangle=-\frac{1}{2}\str (J(X)Y+J(Y)X).\label{form}
\end{equation}
If we write $X$ and $Y$ in the block form as $(X_{ij})$ and $(Y_{ij})$, $i,j=0,1$, then
$$ \langle X, Y\rangle=\frac{1}{2}\sum_{i,j=0}^1(-1)^{ij}\tr\bigl( X_{ij}^\ast Y_{ij}+Y_{ij}^\ast X_{ij}\bigr).$$


In what follows $G$ denotes $S(GL_n(\cx)\times GL_m(\cx))$ for $n\neq m$  and $G=SL_n(\cx)\times SL_m(\cx)$ for $n=m$, and $\fG$ denotes its Lie algebra. 
In order to define an appropriate metric on an orbit of $G$ in $\sM_1=\Mat_{n,m}(\cx)\oplus \Mat_{m,n}(\cx)$ we adopt the following definition.
\begin{definition} An element $C$ of $\sM_1$ is called {\em $ \langle \: ,\rangle$-regular} if $\Ker \ad C\subset \fG$ is nondegenerate with respect to the form \eqref{form}.
\end{definition}

If $C$ is $ \langle \: ,\rangle$-regular, then $\Ker \ad C$ has  an $ \langle \: ,\rangle$-orthogonal complement $V_C$, which is also  $ \langle \: ,\rangle$-nondegenerate. In this case, we can decompose uniquely any $X\in \fG$ as $X=X^C+X^0$ with $X^C\in V_C$ and $X^0\in \Ker \ad C$.
\par
Let now $\sO$ be an orbit of $G$ in $\sM_1$ and $C\in \sO$ its $J$-regular element.   For  two  vectors $[C,X]$ and $[C,Y]$ tangent to $\sO$ at $C$ we define their inner product to be $\langle X^C, Y^C\rangle$. We obtain a pseudo-Riemannian  metric on the $J$-regular part of $\sO$, which we denote by $\langle\:,\rangle_\sO$.

\begin{theorem} On the $ \langle \: ,\rangle$-regular part of $\sO$ the flow \eqref{sJC} is the gradient flow of the  function $H(C)=\frac{1}{4}\langle C,C\rangle$  with respect to the metric $ \langle\:,\rangle_\sO$.
\end{theorem}
\begin{proof}
 The function $H$ can be written as $-\frac{1}{4}\str J(C)C$. By the definition of the gradient we have, for any tangent vector $[C,\rho]_s$,
$$ \langle \grad H,[C,\rho]_s\rangle_\sO=-\frac{1}{4}\str\bigl( J([C,\rho]_s)C+J(C)[C,\rho]_s\bigr)=-\frac{1}{2}\re\str  J(C)[C,\rho]_s.$$
Setting $\grad H=[C,X]_s$, we can rewrite this as
$$\langle X^C,\rho^C\rangle =\frac{1}{2}\re\str  J(C)[C,\rho]_s.$$
Recalling \eqref{ad-inv}  and using the fact that $|C|=|J(C)|=1$
we obtain
$$\langle X^C,\rho^C\rangle=\frac{1}{2}\re\str [C,J(C)]_s\rho= - \frac{1}{2} \langle [J(C),C]_s,\rho\rangle.$$
Since $\Ker \ad C$ is $\langle\:,\rangle$-orthogonal to $\im \ad_s J(C)$, we have $\langle  [J(C),C]_s,\rho\rangle=\langle  [J(C),C]_s,\rho^C\rangle$ and  $[J(C),C]_s\in V_C$. Since the metric $\langle\:,\rangle$ is nondegenerate
on $V_C$,  we can conclude that
$X^C=-\frac{1}{2}[J(C),C]_s$. 
Thus
$\grad H=[C,X]_s=[C,X^C]_s=\frac{1}{2}[[J(C),C]_s,C]_s$.
\end{proof}

\section{Nahm's equations from anti-Lie triple systems}


As observed in the previous section, the BHT-equation \eqref{Ter} have a natural interpretation as a double superbracket equation on the odd part of the Lie superalgebra $\gl_{n|m}(\cx)$. We shall now generalise this to arbitrary Lie superalgebras, or, equivalently to the anti-Lie triple systems \cite{FoF}.

An {\em anti-Lie triple system} (ALTS) is a vector space with a triple (trilinear) product $[\cdot, \cdot,\cdot]$ satisfying the following  identities
\begin{equation}
 \begin{matrix}[x,y,z]=[y,x,z]\\
 [x, y, z] + [z, x, y] + [y, z, x] = 0,\\
 [u, v, [x, y, z]] = [[u, v, x], y, z] + [x, [u, v, y], z] + [x, y, [u, v, z]].
 \end{matrix}
\label{ALTS}
\end{equation}

The third equation can be rewritten as a condition on left multiplications $L(\cdot,\cdot)$ (defined via $L(x,y)z=[x,y,z]$):
\begin{equation} \bigl[L(u,v),L(x,y)\bigr]=L(L(u,v)x,y)+L(x,L(u,v)y).\label{L}
\end{equation}

A basic example is the vector space of  $k \times k$-matrices with the triple product:
\begin{equation} [A,B,C]=ABC+BAC-CBA-BCA \label{AL}\end{equation}
This triple product leaves invariant the subspace 
$\Mat_{n,m}\oplus \Mat_{m,n}$ of off-diagonal blocks, and so the latter is also an ALTS.

We recall  \cite{FoF} the construction of a Lie superalgebra associated to an anti-Lie triple system $(V,[\cdot, \cdot,\cdot])$. 
\par
 Let $D(V)$ denote the Lie algebra of all left multiplications $L(x,y)$ on $(V,[\cdot,\cdot,\cdot])$.
Then $D(V)\oplus V$ becomes a Lie superalgebra $\fL(V)$ under the following bracket:
\begin{equation}
 \begin{matrix}\bigl[L(x,y),L(u,v)\bigr]=L(x,y)\circ L(u,v)- L(u,v)\circ L(x,y)\\
 \bigl[L(x,y),z]=L(x,y)z\\
 [x,y]=L(x,y). 
 \end{matrix}
\label{super}
\end{equation}
The even part of $\fL(V)$ is $\fL_0=D(V)$ and the odd one is $\fL_1=V$. Conversely, given a Lie superalgebra $\fL=\fL_0\oplus \fL_1$, the double superbracket defines an anti-Lie triple product on $\fL_1$:
\begin{equation}[x,y,z]=[[x,y],z].\label{double}
\end{equation}

\begin{example}Applying this construction to the ALTS $\Mat_{n,m}\oplus \Mat_{m,n}$  with the triple product given by \eqref{AL} produces the Lie superalgebra $\fG\fL_{n|m}(\cx)$.
\end{example}

Let now $(V,[\cdot, \cdot,\cdot])$ be a complex ALTS and $J$ a quaternionic automorphism, i.e. $J$ preserves the triple product, is antilinear, and satisfies $J^2=-1$.  We can extend $J$ to an antilinear automorphism of $\fL(V)$ by setting $J(L(x,y))=L(J(x),J(y))$. On $\fL_0$ it satisfies $J^2=1$, an so the Lie algebra $\fL_0$ has a symmetric pair decomposition $\fL_0=\fK\oplus \fM$, where $\fK$ is the $+1$-eigenspace and $\fM$ the $-1$-eigenspace of $J$. The antilinearity of $J$ implies that,  the following three functions
 \begin{equation*} T_1=-\frac{i}{2}[C,J(C)]\quad T_2=\frac{1}{2}[C,C]_\fK\quad T_3=-\frac{i}{2}[C,C]_\fM\end{equation*}
 take values in $\fK$.\\
 We consider the following ODE on $V$:
\begin{equation} \dot{C}=\frac{1}{2}[J(C),C,C],\label{JCCC}
\end{equation}
\begin{proposition} $C=C(t)$ is a solution of \eqref{JCCC} if and only if $T_1,T_2,T_3$ satisfy the Nahm equations.
\end{proposition}
\begin{proof}
The definition implies that $T_2=\frac{1}{4}([C,C]+J[C,C])$ and $T_3=-\frac{i}{4}([C,C]-J[C,C])$. Setting $\alpha=iT_1$ and $\beta=T_2+iT_3$, we have
$$ \alpha=\frac{1}{2}[C,J(C)],\quad \beta=\frac{1}{2}[C,C].$$
The super-Jacobi identity implies that if $x$ is an odd element of a Lie superalgebra $\g$, then for any $y\in \g$
\begin{equation}
 [x,[x,y]]=\frac{1}{2}[[x,x],y].\label{same}
\end{equation}
 In particular, the equation \eqref{JCCC} can be rewritten as 
$$ \dot{C}=\frac{1}{4}[[J(C),[C,C]].$$
We compute using \eqref{same}:
\begin{multline*}
 \dot\alpha=\frac{1}{2}[\dot{C},J(C]]+\frac{1}{2}[C,J(\dot{C})]=\frac{1}{8}[[[J(C),[C,C]],J(C)]-\frac{1}{8}[C,[C,[J(C),J(C)]]\\
=\frac{1}{8}[[J(C),J(C)],[C,C]]=\frac{1}{2}[J(\beta),\beta]=\frac{1}{2}[T_2-iT_3,T_2+iT_3]=i[T_2,T_3].\cr
\end{multline*}
Similarly:
$$
 \dot\beta=[\dot C,C]=\frac{1}{2}[[[J(C),C],C]],C]=\frac{1}{4}[[J(C),C],[C,C]]=[\alpha,\beta],
$$
which is equivalent to the remaining two Nahm equations.
\end{proof}

\begin{remark} We can also consider an arbitrary, real or complex, anti-Lie triple system $(V, [\cdot,\cdot,\cdot])$ equipped with an automorphism $J$ such that $J^2=-1$. The even part of the Lie algebra $\fL_0$ still has the symmetric decomposition $\fL_0=\fK\oplus\fM$ into the $\pm$-eigenspaces of $J$ extended to $\fL_0$. We can define the three functions: 
$$R_1=\frac{1}{2}[C,J(C)]\in \fM,\quad R_2=\frac{1}{2}[C,C]_\fK\in \fK,\quad R_3=\frac{1}{2}[C,C]_\fM\in \fM.$$
Equation \ref{JCCC} implies that $R_1,R_2,R_3$ satisfy the {\em Nahm-Schmid} equations \cite{Schmid}:
$$ \dot{R}_1=\frac{1}{2}[R_2,R_3],\quad \dot{R}_2=\frac{1}{2}[R_1,R_3],\quad \dot{R}_3=\frac{1}{2}[R_1,R_2].
$$ 
\end{remark}

\section{Flows on Jacobians}

It is well-known \cite{Hit1,Hit2} that Nahm's equations correspond to a linear flow on the Jacobian\footnote{This is true if the curve is smooth or integral; in general, the flow is on the generalised Jacobian or on the moduli space of higher rank vector bundles.} of an algebraic curve embedded in $T\oP^1$, i.e. in the total space $|\sO(2)|$ of the line bundle $\sO_{\oP^1}(2)$. Similarly, as we shall shortly see, the Basu-Harvey-Terashima equations \eqref{3rd} correspond to a linear flow on the {\em equivariant} Jacobian of a curve in $\oP^2\backslash \oP^1$, i.e. in the total space of the line bundle $\sO_{\oP^1}(1)$.
\par
In this section we aim to make precise the correspondence between the Nahm flow and the Basu-Harvey-Terashima flow on the Jacobians. We consider first the purely holomorphic picture in the spirit of Beauville \cite{Beau1}. Thus the Nahm matrices are replaced by a quadratic matrix polynomial $X(\zeta)=X_0+X_1\zeta+X_2\zeta^2$ with $X_i\in \gl_n(\cx)$. Such a polynomial corresponds to an acyclic $1$-dimensional sheaf $\sF$ on $T=|\sO(2)|$ defined via
\begin{equation}
 0\to \sO_T(-3)^{\oplus n}\stackrel{\eta-X(\zeta)}{\longrightarrow}\sO_T(-1)^{\oplus n}\longrightarrow \sF\to 0. \label{res1}
\end{equation}
The support of $\sF$ is the $1$-dimensional scheme $S$ cut out by $\det(\eta-X(\zeta))$. As long as this polynomial is irreducible, $S$ is integral and $\sF$ is a line bundle on $S$. More generally, $\sF$ is a line bundle (i.e. an invertible sheaf) on $S$ as long as $X(\zeta)$ is a regular element of $\gl_n(\cx)$ for each $\zeta\in \oP^1$ (with $X(\infty)=X_2$). In fact, we have the following result of Beauville:
\begin{theorem}[Beauville \cite{Beau1}] Let $d$ be a positive integer and  $P(\zeta,\lambda)=\lambda^k+a_1(\zeta)\lambda^{k-1}+\dots+a_k(\zeta)$ a polynomial with $\deg a_i(\zeta)=id$, $i=1,\dots,k$. Consider the variety
$$M(P)=\{X(\zeta)\in \gl_n(\cx)[\zeta]\;;\; \deg X(\zeta)=d,\enskip\det(\lambda-X(\zeta))=P(\zeta,\lambda)\},$$
and its subvariety $M(P)^{\rm reg}$ consisting of $X(\zeta)$ which are regular for each $\zeta\in \oP^1$.
\par
 The following exact sequence on $T=|\sO(d)|$
\begin{equation}
 0\to \sO_T(-d-1)^{\oplus k}\stackrel{\lambda-X(\zeta)}{\longrightarrow}\sO_T(-1)^{\oplus k}\longrightarrow \sF\to 0 \label{resd}
\end{equation}
 induces a $1-1$ correspondence between $M(P)^{\rm reg}/GL_k(\cx)$ and $\Jac^{g-1}(S)-\Theta$, where $S\subset |\sO(d)|$ is the curve of (arithmetic) genus $g=(k-1)(dk-2)/2$ defined by the equation $P(\zeta,\lambda)=0$.\label{Be}
\end{theorem}
We shall call this correspondence the {\em Beauville isomorphism}.

\subsection{$\tau$-sheaves}
We now replace $W_{n,m}$ by its complexification, i.e. the vector space $R_{n,m}$ of quadruples of complex matrices $(A_0,A_1,B_0,B_1)$ with $A_0,A_1$ of size $n\times m$, $B_0,B_1$ of size $m\times n$.  $R_{n,m}$  a {\em biquaternionic} vector space, i.e. a module over $\Mat_{2,2}(\cx)$, and comes equipped with a $2$-sphere of complex symplectic structures:
\begin{equation}
 \omega:\tr d(A_0+A_1\zeta)\wedge d(B_0+B_1\zeta),
\end{equation}
where $\zeta$ denotes the affine coordinate on $\oP^1$.  We can view  $\omega$ itself  as an $\sO(2)$-twisted symplectic form. It is clearly $GL(m,\cx)\times GL(n,\cx)$-invariant. The (twisted) moment map for the $GL(n)$-action is given by:
\begin{equation} \mu:(A_0,A_1,B_0,B_1)\longmapsto A_0B_0 + (A_0B_1+A_1B_0)\zeta + A_1B_1\zeta^2,\label{nn}
\end{equation}
while the one for the $GL(m)$-action is:
\begin{equation} \nu:(A_0,A_1,B_0,B_1)\longmapsto-B_0A_0-(B_0A_1+B_1A_0)\zeta-B_1A_1\zeta^2.\label{mm}
\end{equation}
These are complexifications of the moment maps defined in \S\ref{1}. As in section \ref{Lax} we can view $R_{n,m}$ as the following subset of $\gl(m+n)\otimes \cx^2$:
\begin{equation}
 C_0=\begin{pmatrix} 0 & A_0\\ B_0 & 0\end{pmatrix},\enskip C_1=\begin{pmatrix} 0 & A_1\\ B_1 & 0\end{pmatrix}.\label{C}
\end{equation}
For any pair $C_0,C_1$ of quadratic matrices, say of size $k\times k$, we can define an acyclic $1$-dimensional sheaf on $\hat{T}=|\sO(1)|$ via the exact sequence \eqref{resd} with $d=1$ and $X(\zeta)=C_0+C_1\zeta$. We are interested in the structure of these sheaves and their supports for $C_0,C_1$ of the form \eqref{C}, and in their relation to sheaves on $|\sO(2)|$ defined via maps \eqref{nn} and \eqref{mm}. Observe that $|\sO(2)|$ is the quotient of $|\sO(1)|$ by the following involution on $|\sO(1)|$:
\begin{equation}
 \tau(\zeta,\lambda)=(\zeta,-\lambda)\label{tau}.
\end{equation}
For an element $(A_0,A_1,B_0,B_1)$ of $R_{n,m}$ with $n\geq m$, the polynomial $\det(\lambda-C_0-C_1\zeta))$ is $\tau$-invariant and of the form
\begin{equation}
 P(\zeta,\lambda)=\lambda^{n-m}\bigl(\lambda^{2m}+a_1(\zeta)\lambda^{2m-2}+\dots+a_{m-1}(\zeta)\lambda^2+a_m(\zeta)\bigr),\quad \deg a_i(\zeta)=2i.  
\label{square}\end{equation}
$GL_n(\cx)\times GL_m(\cx)$-orbits of elements of $R_{n,m}$ correspond  to acyclic {\em $\tau$-sheaves} on $\hat{S}=\{(\zeta,\lambda); P(\zeta,\lambda)=0\}$,
i.e. sheaves equivariant with respect to the action of $\tau$. In the case of a line (or vector) bundle $\sF$ on $\hat{S}$ this means that $\tau$ lifts to an involutive bundle map on the total space of $\sF$.
\par
Let us write $R_{n,m}(P)$ for $R_{n,m}\cap M(P)$ and $R_{n,m}(P)^{\rm reg}$ for $R_{n,m}\cap M(P)^{\rm reg}$. We have
\begin{proposition}
 The Beauville isomorphism induces  a $1-1$ correspondence between $R_{n,m}(P)^{\rm reg}/GL_n(\cx)\times GL_m(\cx)$ and the isomorphism classes of acyclic $\tau$-line bundles on $\hat{S}$.
\end{proposition}
\begin{proof}
 Let  $(A_0,A_1,B_0,B_1)$ with the corresponding $C(\zeta)=C_0+C_1\zeta$ given by \eqref{C} belong to $R_{n,m}(P)^{\rm reg} $. Then $g_0C(\zeta)g_0^{-1}=-C(\zeta)$, where $g_0=\begin{pmatrix} {\rm Id}_n & 0\\ 0 & -{\rm Id}_m\end{pmatrix}$. The commutative diagram on $|\sO(1)|$ (with $k=n+m$)
$$\begin{CD}
0 @>>> \sO(-2)^{\oplus k} @>{\lambda-C(\zeta)}>>\sO(-1)^{\oplus k} @>>> \sF@>>> 0\\
@. @Vg_0VV @Vg_0VV @V\tilde{\tau}VV @.\\
0 @>>> \sO(-2)^{\oplus k} @>{\lambda+C(\zeta)}>>\sO(-1)^{\oplus k} @>>> \sF@>>> 0
\end{CD}
$$
defines a lift $\tilde{\tau}:\sF\to \sF$ of $\tau$. Conjugating $ C(\zeta)$ by an element $T$ of $GL_n(\cx)\times GL_m(\cx)$ commutes with $g_0$ and so
$C(\zeta)$ and $TC(\zeta)T^{-1}$ induce isomorphic $\tau$-sheaves. Conversely, suppose that we are given a lift $\tilde{\tau}$ of $\tau$ on an acyclic line bundle $\sF$, satisfying $\tilde{\tau}^2=1$. We obtain the corresponding involution $\tilde{\tau}$ on $\cx^{n+m}=H^0(\hat{S},\sF(1))$. We can choose a basis of $H^0(\hat{S},\sF(1))$ so that $\tilde{\tau}$ is represented by the matrix $g_0=\begin{pmatrix} {\rm Id}_n & 0\\ 0 & -{\rm Id}_m\end{pmatrix}$. It follows from the above commutative diagram that  $g_0C(\zeta)g_0^{-1}=-C(\zeta)$, so that $(A_0,A_1,B_0,B_1)$ belongs to $R_{n,m}$.
\end{proof}

\subsection{The case $n=m$} In this case the quotient of the curve $\hat{S}$ by the involution $\tau$ is (as a scheme) a curve $S$ in $T=|\sO(2)|$.  The maps \eqref{nn} and \eqref{mm} induce, via the above Proposition and Theorem \ref{Be}, correspondences between acyclic  $\tau$-line bundles on $\hat{S}$ and acyclic line bundles on $S$. We wish to understand these correspondences.
\par
We shall write $A(\zeta)$ for $A_0+A_1\zeta$ and $B(\zeta)$ for $B_0+B_1\zeta$, so that the map $\mu$ gives the quadratic matrix polynomial $A(\zeta)B(\zeta)$ and $\nu$ the polynomial $B(\zeta)A(\zeta)$. Let us write $\hat{P}(\zeta,\lambda)$ for the polynomial $\det(\lambda-C(\zeta))$ and $P(\zeta,\eta)$ for the polynomial $\det(\eta-A(\zeta)B(\zeta))=\det(\eta-B(\zeta)A(\zeta))$. Denote by $\hat{S}$ the curve in $\oP^2$ cut out by $\hat{P}$ and by $S$ the curve cut out by $P$ in $T\oP^1$. We have
$$ \hat{P}(\zeta,\lambda)=P(\zeta,\lambda^2),$$
so that $\hat{S}$ is a double cover of $S$ ramified over $\eta=0$. The genus of $S$ is equal to $(n-1)^2$ and the genus of $\hat S$ is equal to $(n-1)(2n-1)$.
\par
We shall denote by $\sL$ the acyclic $\tau$-sheaf on $\hat{S}$ defined by $C(\zeta)$ and by $\sF$ (resp. $\sG$) the acyclic sheaf on $S$ defined by $A(\zeta)B(\zeta)$ (resp. $B(\zeta)A(\zeta)$).  We shall assume that the zeros of $\det A(\zeta)$ are distinct from the zeros of $\det B(\zeta)$, and we shall write $\Delta_A$ (resp. $\Delta_B$) for the divisor $\det A(\zeta)=0$, $\lambda=0$ (resp. $\det B(\zeta)=0$, $\lambda=0$) on $\hat{S}$. Thus $\Delta_A+\Delta_B$ is the ramification divisor of the projection $\tau$.

\begin{proposition} With the above assumptions $\sL\simeq \pi^\ast \sF\otimes [\Delta_B]\simeq \pi^\ast \sG\otimes [\Delta_A]$, where $\pi:\hat S\to S$ is the projection.
\end{proposition}
\begin{proof} We consider the sheaves $\sL(1)$, $\sF(1)$ and $\sG(1)$ which are cokernels of $\lambda-C(\zeta): \sO(-1)^{\oplus 2n}\to \sO^{\oplus 2n}$,
$\eta-A(\zeta)B(\zeta): \sO(-2)^{\oplus n}\to \sO^{\oplus n}$, and of $\eta-B(\zeta)A(\zeta): \sO(-2)^{\oplus n}\to \sO^{\oplus n}$, respectively.
Any vector $u\in \cx^n$ defines  a global section $s_u$ of $\sF(1)$ via \eqref{res1}. We choose $u$ so that the zeros of $s_u$  are disjoint from $\eta=0$ and from the singular locus of $S$. In other words $u\not\in \im A(\zeta)B(\zeta)$ if $\det A(\zeta)B(\zeta)=0$ and $u\not\in \im (\eta-A(\zeta)B(\zeta))$ for a singular point $(\zeta,\eta)\in S$.  Consider the vector $(u,0)\in \cx^{2n}$ which defines a global section $\hat{s}_u$ of $\sL(1)$. It is then easy to check that $(u,0)\in \im (\lambda-C(\zeta))$ if either $\lambda\neq 0$ and $u\in \im(\lambda^2-A(\zeta)B(\zeta))$ or $\lambda=0$, $\det B(\zeta)=0$ and $u\in\im A(\zeta)$. Since $\Delta_A$ and $\Delta_B$ are assumed to be disjoint, the condition $u\in\im A(\zeta)$ follows from $\lambda=0$ and $\det B(\zeta)=0$. Thus the divisor $(\hat{s}_u)$ of $\hat{s}_u$ is $\pi^{-1}(s_u)+\Delta_B$ and the first isomorphism follows. The proof of $\sL\simeq \pi^\ast \sG\otimes [\Delta_A]$ is completely analogous.
\end{proof}

\subsection{Flows} The Beauville correspondence implies that the flow of matrices satisfying the Nahm equations corresponds to a flow on $J^{g-1}(S)-\Theta$. It is well-known \cite{Hit1, Hit2} that this latter flow is the linear flow $\sF\mapsto \sF\otimes L^t$, where $L$ is the line bundle with transition function $\exp(\eta/\zeta)$. Similarly the BHT-flow corresponds to a linear flow on the moduli space of acyclic $\tau$-line bundles on $\hat S$ 
in the direction of the line bundle with transition function $\exp(\lambda^2/\zeta)$. In addition, in order to obtain the Basu-Harvey-Terashima equations, rather than purely holomorphic equations \eqref{2nd}, one needs to restrict the flow further to $\sigma$-line bundles on $\hat S$, i.e. line bundles equipped with a lift of the quaternionic structure of $|\sO_{\oP^1}(1)|$.

\end{document}